\documentclass[reprint, 10pt, floatfix,aps,prd,twocolumn,a4paper,superscriptaddress,nofootinbib,notitlepage]{revtex4-1}
\usepackage[utf8]{inputenc}
\usepackage[T1]{fontenc} 
\usepackage{paralist}
\usepackage{hyperref}
\usepackage[english]{babel}
\usepackage{amsopn,amsthm,dsfont,amssymb}
\usepackage{mathrsfs}
\usepackage{cancel} 
\usepackage{comment} 
\usepackage{color} 
\usepackage{soul} 
\usepackage[normalem]{ulem}
\usepackage{cancel} 
\usepackage{hyperref} 
\hypersetup{
    colorlinks=true,
    linkcolor=blue,
    filecolor=magenta,      
    urlcolor=cyan,
}

\AtBeginDocument{%
    \newwrite\bibnotes
    \def\bibnotesext{Notes.bib}
    \immediate\openout\bibnotes=\jobname\bibnotesext
    \immediate\write\bibnotes{@CONTROL{REVTEX41Control}}
    \immediate\write\bibnotes{@CONTROL{%
    apsrev41Control,author="08",editor="1",pages="1",title="0",year="1"}}
     \if@filesw
     \immediate\write\@auxout{\string\citation{apsrev41Control}}%
    \fi
}%

\usepackage{mathtools}
\usepackage{adjustbox}
\usepackage[caption=false]{subfig}

\newtheorem{theorem}{Theorem}
\newtheorem{lemma}{Lemma}
\newtheorem{corollary}{Corollary}
\newtheorem{definition}{Definition}

\newtheorem*{blank*}{}

\begin{document}
\title{Reversible dynamics with closed time-like curves and freedom of choice}

\author{Germain Tobar}
\affiliation{School of Mathematics and Physics,\\The University of Queensland, St.\ Lucia, QLD 4072, Australia}

\author{Fabio Costa}
\affiliation{Centre for Engineered Quantum Systems, School of Mathematics and Physics,\\The University of Queensland, St.\ Lucia, QLD 4072, Australia}

\begin{abstract}
The theory of general relativity predicts the existence of closed time-like curves (CTCs), which theoretically would allow an observer to travel back in time and interact with their past self. This raises the question of whether this could create a grandfather paradox, in which the observer interacts in such a way to prevent their own time travel. Previous research has proposed a framework for deterministic, reversible, dynamics compatible with non-trivial time travel, where observers in distinct regions of spacetime can perform arbitrary local operations with no contradiction arising. However, only scenarios with up to three regions have been fully characterised, revealing only one type of process where the observers can verify to both be in the past and future of each other. Here we extend this characterisation to an arbitrary number of regions and find that there exist several inequivalent processes that can only arise due to non-trivial time travel. This supports the view that complex dynamics is possible in the presence of CTCs, compatible with free choice of local operations and free of inconsistencies.
\end{abstract}

\maketitle

\section{Introduction}
The dominant paradigm in physics relies on the idea that systems evolve through time according to dynamical laws, with the state at a given time determining the entire history of the system. 

General relativity challenges this view. The Einstein equations, describing the relationship between spacetime geometry and mass-energy~\cite{WeinbergGR}, have counterintuitive solutions containing closed time-like curves (CTCs)~\cite{Lanczos:1924kn, Godel:1949eb, Taub:1951, NUT:1963, Kerr:1963, Tipler:1974iw, PhysRevLett.66.1126.1991, Alcubierre_1994, PhysRevD.56.2100.1997, PhysRevD.57.4760.1998, PhysRevD.76.044002.2007, Griffiths:2009exactBook,  Sarma2013, Tippett2017,  Fermi_2018,   Mallary_2018}. An event on such a curve would be both in the future and in the past of itself, preventing an ordinary formulation of dynamics according to an ``initial condition'' problem. The question then arises whether some more general type of dynamics is possible.

Although it is an open question whether CTCs are possible in our universe~\cite{Morris:1988gg, novikov1989analysis, Hawking1992, Visser:1995wormholesBook, Ori2005}, considering dynamics beyond the ordinary temporal view is relevant to other research areas as well. In a theory that combines quantum physics with general relativity, it is expected that spacetime loses its classical properties~\cite{gibbs1995small, Piazza2010}, possibly leading to indefinite causal structures~\cite{hardy2007towards, HardyProbabilityGravity, Zych2019}. In a quite different direction, it has been suggested that quantum physics could be reduced to some kind of ``retrocausal'' classical dynamics \cite{carati1999nonlocality, Weinstein2009, price2012, priceWharton2015, Evans01062013, Wharton2014, Aharonov2016, Leifer2016, Sutherland17, Shrapnel2018causationdoesnot, Wharton2018, Emily2018}.

The main problem arising when abandoning ordinary causality is the so called ``grand father paradox''~\cite{schachner33}: a time traveller could kill her own grandfather and thus prevent her own birth, leading to a logical inconsistency. A popular approach holds that the grandfather paradox makes CTCs incompatible with classical physics, while appropriate modifications to quantum physics could restore consistency~\cite{Deutsch:1991jo, Politzer1994, Pegg:2001wa, Bennett, Bacon2004, Greenberger2005, Svetlichny:2009ve,Svetlichny:2011gq, Lloyd2011exp, Lloyd:2011ir, Ralph2010,Pienaar2011, Ralph:2012cd, Wallman2012, Allen2014, Czachor}. A common feature of the proposals within this approach is that they postulate a radical departure from ordinary physics even in regions of space-time devoid of CTCs, or in scenarios where the time travelling system does not actually interact with anything in the past~\cite{Pienaar2013, Yuan2015}.

A different approach is the so called ``process matrix formalism'', which takes as a starting point the local validity of the ordinary laws of physics and asks what type of global processes are compatible with this assumption~\cite{Oreshkov:2012uh, araujo15, feixquantum2015, Oreshkov2015, Branciard2016, Giacomini2015, Baumann2016, Costa2016, abbott2016, araujo2017purification, Perinotti2017, Abbott2017genuinely, Shrapnel2018, Jia2018, Bavaresco2019semidevice, Jia2019}. This framework enforces that all operations that would normally be possible in ordinary spacetime should still be available in local regions.  First considered in the quantum context, this approach has been applied to classical physics too, with the remarkable discovery of classical processes that are incompatible with any causal order between events~\cite{baumeler14, Baumeler2015, Baumeler2016}. 


In Ref.~\cite{Baumeler2019}, a classical, deterministic version of the formalism was proposed as a possible model for CTCs. In this model, one considers a set of regions that do not contain any, but might be traversed by, CTCs. Agents in the regions receive a classical state from the past boundary, perform an arbitrary deterministic operation on it, and then send the system through the future boundary. Dynamics outside the regions determines the state each agent will observe in the past of the respective region, as a function of the states prepared by other agents. A simple characterisation was found for all processes involving up to three regions; furthermore, it was found that, for three regions, all non causally ordered processes are essentially equivalent. 

In this work, we extend the characterisation of deterministic processes to an arbitrary number of regions. We provide some simple interpretation of the characterisation: when fixing the state on the future of all but two regions, the remaining two must be causally ordered, with only one directional signalling possible. 
We show, by explicit examples, that there are inequivalent, non causally ordered quadripartite processes, which cannot be reduced to tripartite ones. Our results show that CTCs are not only compatible with determinism and with the local ``free choice'' of operations, but also with a rich and diverse range of scenarios and dynamical processes. 

\section{Deterministic processes}
This section aims to revise and summarise the approach of Ref.~\cite{Baumeler2019} and the results that are relevant to the full characterisation of arbitrary deterministic, classical processes.

In ordinary dynamics, a \textit{process} is a function that maps the state of a system at a given time to the state at a future time. Operationally, we can think of the state in the past as a `preparation' and the one in the future as the outcome of a `measurement'. The functional relation between preparation and measurement is typically dictated by the dynamics of the system. For example, field equations determine the field on a future spacelike surface as a function of the field on a past spacelike surface. In a globally hyperbolic spacetime\footnote{A spacetime is said to be globally hyperbolic if and only if it contains a Cauchy surface, namely a closed, acausal surface such that every inextendible causal curve passes through it once (where acausal means that no causal curve starts and ends on the surface) \cite{Minguzzi2019}.}, the functional relation between preparation and measurement corresponds to fixed deterministic dynamics once initial conditions have been set on a Cauchy surface.

The first studies of dynamics in the presence of CTCs focused on spacetime geometries in which a formulation of dynamics as an initial condition problem could be retained \cite{Friedman:1990ja, Echeverria:1991ko, Lossev1992}. In these studies, CTCs only exist in the future of a spacelike surface on which initial conditions could be set. Surprisingly, these studies found that all cases studied had at least one self-consistent solution. This surprising result suggests there may not be any conflict between the existence of CTCs and logical consistency. However, if the spacetime is threaded by CTCs, it is not in general possible to find a spacelike surface to set global initial conditions.

The difficulty associated with setting global initial conditions in spacetimes threaded by CTCs raises the question as to whether there exists a more general description of a process which can describe dynamics without global initial conditions. The existence of such processes could be used to model dynamics incompatible with the existence of a Cauchy surface. Of course, the non-existence of a Cauchy surface (and therefore the spacetime being non-globally hyperbolic), does not immediately imply the presence of CTCs. However, the development of such a generalised type of process in Ref.~\cite{Baumeler2019} provides a framework to describe dynamics compatible with non-trivial time travel in a non-globally hyperbolic spacetime. The remainder of this section reviews how such a generalised process can be constructed to distinguish whether the process is compatible with spacetimes containing CTCs (or more generally non-trivial time travel in a non-globally hyperbolic spactime), while remaining consistent with free choice, locality and the absence of a grandfather paradox.

In order to generalise the ordinary characterisation of a process as a function with initial conditions, we consider $N$ spacetime regions (which we henceforth refer to as \textit{local regions} or \textit{regions}), in which agents can perform arbitrary operations. In particular, each agent will observe a state coming from the past of the region and prepare a state to send out through the future. The key assumption is that the actions of the agents in the regions are independent from the relevant dynamics governing the exterior of the regions. In other words, agents retain their ``freedom of choice'' to perform arbitrary operations. However, the system's dynamics, together with the spacetime's geometry, constrains the system's behaviour outside the regions. Therefore outside the spacetime regions, the dynamics is completely deterministic and fixed, once appropriate boundary conditions have been specified. In this approach, a \textit{process} should determine the outcomes of measurements performed by an agent, as a function of the operations performed by the others. In this way, we define a \textit{process} as a generalised model for the dynamics connecting distinct spacetime regions. 

We shall assume that the regions in which the agents act are connected, non-overlapping, and their boundaries can be divided into two subsets (one ``past'' and one ``future'') that do not contain timelike pieces. This ensures that, in a CTC-free spacetime (or more generally, a globally hyperbolic spacetime), each region is either in the future, in the past, or spacelike to any other, so a violation of causal order between the regions can be attributed to a lack of causal order in the background spacetime\footnote{If we were to consider timelike boundaries, it would be possible to have time-like curves which exit a region, cross into another one, and then go back to the original region, all within a globally hyperbolic spacetime.}.
Furthermore, we shall restrict to local regions that do not contain CTCs, or more generally we require that these localised regions are indistinguishable from localised regions in a globally hyperbolic spacetime. We also assume that any  timelike curve which enters through the past (future boundary) exits through the future boundary (past boundary). This enables a simple characterisation of local operations in the regions as functions from past to future boundary.

Besides the above conditions, we do not make any further assumption on the spacetime geometry or causal structure, nor on the dynamics of the physical system on which the agents act. The approach only treats abstractly the functional relations between degrees of freedom and only requires that local agents retain their freedom to choose arbitrary operations, without generating any logical inconsistency.

Our goal is to understand whether such abstract constraints are compatible with processes that can only be realised through non-trivial causal relations (causal relations that can only arise due to non-trivial time travel in a non-globally hyperbolic spacetime). By causal relations between regions, here we mean the possibility for agents acting in the regions to exchange non-faster-than-light signals between each other. In a globally hyperbolic spacetime (and under the assumptions we have made on the local regions), signalling is a one-way relation: if an agent $A$ can signal to an agent $B$, then $B$ cannot signal to $A$. Therefore, any process where causal relations do not define a partial order among the regions would be due to non-trivial time travel in a non-globally hyperbolic spacetime.

In keeping with previous literature, and to simplify the discussion, in the following we assume that CTCs are responsible for non-trivial causal relations, such that the local regions do not contain any, but may be traversed by CTCs. However, these non-trivial causal relations could in principle emerge in other non-globally hyperbolic spacetimes without CTCs, such as causal but not strongly causal spacetimes \cite{Minguzzi2019}.

\subsection{The Process Function}
In order to develop the formalism for classical, deterministic dynamics of local regions in the presence of CTCs, we will assign the boundaries of these local regions classical state spaces. The state spaces $\mathcal{A}_i$ and $\mathcal{X}_i$ describe the physical degrees of freedom localised on the past and future boundaries respectively of a local region $i$. For example, they may correspond to a field defined on the background spacetime. In this case, a state in one of the classical state spaces we have defined, would be a function on the boundary of the local region, describing a spacelike field configuration. Individual states will be denoted as $a_i \in \mathcal{A}_i$, $x_i \in \mathcal{X}_i$. A classical, deterministic operation in the local region will be denoted by the function $f_i: \mathcal{A}_i \rightarrow \mathcal{X}_i$ (Fig.~\ref{function}). The function $f_i$ transforms the input state $a_i$ from the past boundary to the output state $x_i$ at the future boundary. Therefore, the function $f_i$ maps between the physical degrees of freedom on the past and future boundaries respectively. The function $f_i$ physically corresponds to the local operation an agent can perform on the input state in a particular spacetime region. Therefore, in order to allow freedom of choice over the operations that an agent can perform in a local region of spacetime, we do not impose any additional condition on the function $f_i$, apart from requiring that it satisfies the definition of a function. As a result, the local functions $f_i$ are not required to be invertible. Allowing the local operations performed by experimenters to be non-invertible functions represents the experimenters' ability to delete information by accessing a reservoir not included in the physical degrees of freedom of interest. We denote $\mathcal{D}_i := \{f_i: \mathcal{A}_i \rightarrow \mathcal{X}_i \}$ to be the set of all possible functions in region $i$. In order to refer to a collection of objects for all regions, we drop the index. For example, the set of all possible inputs for $N$ distinct local regions will be denoted $\mathcal{A} \equiv \mathcal{A}_1 \times . . . \times \mathcal{A}_N$.

We will use the notation $\mathcal{A}_{\setminus i}= \mathcal{A}_1\times\dots\times\mathcal{A}_{i-1}\times \mathcal{A}_{i+1}\times\dots\times\mathcal{A}_N$,  $a_{\setminus i}=\left\{a_1,\dots a_{i-1},a_{i+1},\dots,a_N\right\}$, {\it etc.}, to denote collections with the component $i$ removed. Appropriate reordering will be understood when joining variables, for example in expressions as $a=a_i\cup a_{\setminus i}$ {, $f(a)=f(a_i,a_{\setminus i})$, and so on}.

\begin{figure}[h] 
\begin{center}
\includegraphics[scale  = 0.5]{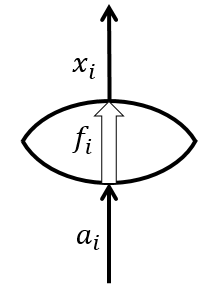}
    \caption{\label{function} A local bounded region $i$ of spacetime with a past and a future boundary. The boundaries do not contain time-like pieces. A local operation in the region is represented by a function $f_i$, which maps the input state $a_i$ containing the physical degrees of freedom on the past bundary, to the output state $x_i$ containing the physical degrees of freedom on the future boundary.}
\end{center}
\end{figure}
One of the requirements of a deterministic framework for local regions in the presence of CTCs, is that the framework must be able to predict the state on the past boundary of each local region. In the presence of CTCs, the state on the past boundary of each local region can depend on all local operations (In a CTC free spacetime, the state on the past-boundary of a region would only depend on operations in its past). The dependence on local operations can be described with a function $\omega \equiv \{\omega_1,...,\omega_N\}: \mathcal{D} \rightarrow \mathcal{A}$ which maps the local operations performed in each region to the input state on the past boundary of each local region \cite{Baumeler2019}. The function $\omega$ will be henceforth labelled as a \textit{process}.

The function $\omega$ will remain general, only being restricted by a weak form of locality. Locality requires that, once the state at the boundary of a region is fixed, the details of what happens inside the region should not be relevant to the exterior dynamics. The local field equations typically used in physics all satisfy this requirement. In order to formalise this requirement for locality, for every process $\omega$ there must exist an additional function $w: \mathcal{X} \rightarrow \mathcal{A}$ such that\footnote{Ref.~\cite{Baumeler2019} defined deterministic processes through a different but equivalent self-consistency condition.}
\begin{equation}
 \omega\left(f\right)=w\left(f\left(\omega(f)\right)\right) \quad \forall f\in \mathcal{D}\,.
\label{consistency}
\end{equation}
We will henceforth refer to a function that satisfies the above consistency condition as a \textit{process function} (Fig.~\ref{processfn}). The \textit{process function} maps the output state on the future boundaries to the input states on the past boundaries of all regions. In a field theory without the presence of CTCs, which is compatible with a causally ordered spacetime, the process function would describe how the field on the future boundary of a region is mapped to the resulting field on the past boundary of a second local region, where the first region is in the causal past of the second region. In the presence of CTCs, the process function represents more general constraints on the boundary configurations of the field, which now can all potentially depend on each other. In this way, the \textit{process function} generalises the notion of dynamics, so that it applies to spacetimes which contain CTCs, as well as spacetimes that are absent of CTCs. The \textit{process function} can therefore be used identify the causal relations that can arise only in the presence of CTCs, without any reference to CTC solutions of the Einstein field equations. This is significant because it enables us to determine what types of non-trivial causal relations are possible without the development of a logical inconsistency such as a grandfather paradox. The investigation of these causal relations with the process function formalism provides a promising platform to investigate whether these non-trivial causal relations can be achieved in specific spacetime geometries. Hence, the consideration of the causal relations in the process function formalism will provide a step in the direction of addressing how CTCs can be physically realised without any grandfather paradox. 

It has been shown in Ref.~\cite{Baumeler2019}, that a necessary and sufficient condition for a process function is that $w \circ f$ has a unique fixed point for every local operation $f$: 
\begin{equation} \label{condition2}
    \forall f \; \exists! \; a \textrm{ such that }\quad w \circ f(a) = a.
\end{equation}
In the following, we will work with process functions, rather than processes, and use the fixed point condition~\eqref{condition2} as the defining property. 

An important property of process functions is that an observer in a localised region cannot use it to send information back to herself. Intuitively, this prevents paradoxes, such as an agent attempting to warn her past self to avoid a particular event, thus removing the motivation for her to warn her past self. Formally, this means that the input of each local region is independent of the output of the same region. If we consider a single local region, the only process function for this sole local region which satisfies the fixed point condition~\eqref{condition2} is a constant: $w \left(x\right) = a$. Therefore, an agent in a local region being unable to interact with her past is a feature of all process functions. In this way, the absence of a grandfather paradox is a direct result of condition~\eqref{condition2}. This feature of process functions imposes the constraint that each component of a process function $w$ has to be independent of the output of the same region:
\begin{equation} \label{indep}
    w_i(x) = w_i(x_{\setminus i}),
\end{equation}
where, $x_{\setminus i}$ is the set of outputs of all regions except the $i$th region. Note that Eq.~\eqref{indep} implies that a process function can be described in terms of a set of functions $w_1: \mathcal{X}_{\setminus 1} \rightarrow \mathcal{A}_1$, ..., $w_n:\mathcal{X}_{\setminus n} \rightarrow \mathcal{A}_n$.

\begin{figure}[h]
\begin{center}
\includegraphics[scale  = 0.45]{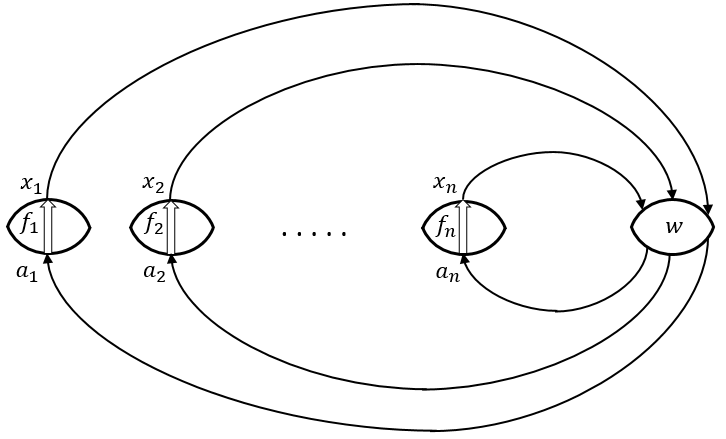}
\caption{ \label{processfn} A process function $w$ can be applied to model the interaction of distinct localised space-time regions with CTCs. A non-trivial dependency of each region's input on the other regions' outputs implies that the process function can only arise in the presence of CTCs.}
\end{center}
\end{figure}

\subsection{Reduced Processes}
Before we go into detail about the characterisation of process functions, we must firstly make some important definitions and consider some important properties of process functions. 
\begin{definition} \label{def1}
Consider a function $w: \mathcal{X} \rightarrow \mathcal{A}$, such that, for each region $i = 1, ....,N$, $w_i(x) = w_i(x_{\setminus i})$. For a particular local operation $f_i: \mathcal{A}_i \rightarrow \mathcal{X}_i$ we define the \textbf{reduced function} $w^{f_i}: \mathcal{X}_{\setminus i} \rightarrow \mathcal{A}_{\setminus i}$ on the remaining regions through the composition of $w$ with $f_i$:

\begin{equation}\label{reduced}	w^{f_i}_j\left(x_{\setminus i}\right):= w_j\left(x_{\setminus i},f_i\left(w_i\left(x_{\setminus i}\right)\right)\right)\,, \; i\not=j\,.
\end{equation}
\end{definition}
This definition is important for formalising the intuition that if we fix the operation for a particular local region, there should still exist a process for the remaining regions. The definition of the reduced function plays an important role in the investigation of the properties of multipartite process functions. 
\begin{lemma}[Lemma 3 in Ref.~\cite{Baumeler2019}] \label{lemma}
	Given a function \mbox{$w:\mathcal{X}\rightarrow \mathcal{A}$}, such that, for each region $i=1,\dots,N$, \mbox{$w_i(x)=w_i(x_{\setminus i})$}, we have
\begin{enumerate}[(i)]
	\item If $w$ is a process function, then $w^{f_i}$ is also a process function for every region $i$ and operation $f_i$.
	\item If there exists a region $i$ such that, for every local operation $f_i$, $w^{f_i}$ is a process function, then $w$ is also a process function.
\end{enumerate}
\end{lemma}
Following the result of this lemma, we can conclude that $w$ is a process function if and only if the corresponding reduced function $w^{f_i}$ is also a process function.

Definition \ref{def1} can be modified to apply to a process in which we fix a particular region's output instead of fixing a particular local operation. These processes are of interest because they correspond to the subset of reduced functions where the fixed operation $f_i$ is a constant. A constant is a valid choice of the function because we do not require it to be invertible.
\begin{definition} \label{def2}
Consider a function $w: \mathcal{X} \rightarrow \mathcal{A}$, such that, for each region $i = 1, ....,N$, $w_i(x) = w_i(x_{\setminus i})$. For a particular region's output $x_i \in \mathcal{X}_i$, we define the \textbf{output reduced function} $w^{x_i}: \mathcal{X}_{\setminus i} \rightarrow \mathcal{A}_{\setminus i}$ on the remaining regions to denote the function in which we have fixed the output of the $i$th region:
\begin{equation}
\begin{split}
    w^{x_i}(x_{\setminus i}) \; : = \; &\{w_1(x_{\setminus 1}),....w_{i-1}(x_{\setminus \{i-1\}}), \\ 
    &w_{i+1}(x_{\setminus \{i+1\}}),... w_n(x_{\setminus n}) \}.
    \end{split}
\end{equation}
\end{definition}

It is clear that an output reduced function is also a reduced function according to definition \ref{def1}, due to the output reduced function corresponding to the case where the local operation in each region is restricted to be a constant. We can apply definition \ref{def2} to denote $w^{x_{\setminus\{ i,j\}}}$ as the output reduced function in which we fixed the outputs of all regions except regions $i$ and $j$. 

While definition \ref{def1} and \ref{def2} are similar, the distinction between the reduced function and the output reduced function is important for the characterisation of multipartite process functions.

\subsection{Signalling} \label{signal}
In order to understand how different parties in distinct regions of spacetime signal to each other we must define what it means for one observer to signal to another observer. 

Equation \eqref{indep} describes how an observer in a local region can not signal to their own past. This is consistent with the following definition of no-signalling.
\begin{definition} \label{nosignalling}
 Given a process function $w:\mathcal{X}\rightarrow \mathcal{A}$, we say that region $j$ cannot signal to region $i$ if
\begin{equation} \label{indep2}
    w_i(x_j,x_{\setminus j}) = w_i(x'_j,x_{\setminus j})\quad \forall\; x \in \mathcal{X},\, x'_j \in \mathcal{X}_j,
    \end{equation}
which we can abbreviate as $w_i(x) = w_i(x_{\setminus j})$.
\end{definition}
We define \textit{signalling} as the negation of definition \ref{nosignalling}.

Signalling is useful to establish whether a process is compatible with a given causal structure, as a region can only signal to regions in its future. This definition of signalling is consistent the aim of constructing the framework such that in a spacetime without CTCs, signalling between regions defines a relation of partial order. However, the presence of CTCs does not automatically allow arbitrary signalling, as the consistency condition \eqref{condition2} imposes strong constraints on the process function.

As we will show below, it is convenient to characterise process functions in terms of a more refined notion of signalling. In general, the possibility to signal from a region to another can depend on the outputs of all other regions. It is useful to capture this as follows:

\begin{definition}
Given two regions $i$ and $j$, a process function $w:\mathcal{X}\rightarrow \mathcal{A}$, and an output state $\tilde{x}_{\setminus \{i,j\}} \in \mathcal{X}_{\setminus \{i,j\}}$, we say that $j$ cannot signal to $i$ conditioned on $\tilde{x}_{\setminus \{i,j\}}$ if 
\begin{equation} \label{signalcon}
    w_i^{\tilde{x}_{\setminus \{i,j\}}}(x_j) = w_i^{\tilde{x}_{\setminus \{i,j\}}}(x'_j)\quad \forall\; x_j, \, x'_j \in \mathcal{X}_j.
\end{equation}
\end{definition}

For example, for certain process functions $w$, there can exist $x_{\setminus \{i,j\}} \in \mathcal{X}_{\setminus \{i,j\}}$ such that each of the components $w_i^{x_{\setminus \{i,j\}}}$ and $w_j^{x_{\setminus \{i,j\}}}$ are a constant. In these cases neither region can signal to the other. However, there may also exist another choice of outputs  $x_{\setminus \{i,j\}}' \in \mathcal{X}_{\setminus \{i,j\}}$ such that signalling occurs between regions $i$ and $j$.

We now have a framework which describes general deterministic dynamics in the presence of CTCs. This framework is characterized by the process function $w$ which maps the output states on the future boundary of each local region to the input states on the past boundary of each local region. Condition \eqref{condition2} allows freedom of choice for the operations performed by the observer in each region. Condition \eqref{indep} guarantees that there is no paradox resulting from the operations performed in the presence of CTCs. In order to further understand communication between observers in the presence of CTCs, we must develop a  characterisation of the process function $w$ which describes how these observers can communicate.

\subsection{Characterization of Process Functions}
The simplest and most intuitive process functions are causally ordered ones. For example, consider three observers in three distinct regions, which we label regions 1, 2 and 3 respectively. If there exists causal order between these regions such that $1 \prec 2 \prec 3$, then the process function is given by $w_1(x) = a$ (constant), $w_2(x) = w_2(x_1)$ and $w_3(x) = w_3(x_1,x_2)$. For such causally ordered process functions condition \eqref{condition2} is satisfied. However, these trivial, causally ordered process functions are compatible with spacetimes that do not contain CTCs. The generality of the process function allows it to model spacetime without causal order as well, such as the non-trivial causal relations which arise due to the presence of CTCs. We are interested in whether non-trivial process functions exist in the presence of CTCs. The existence of non-trivial process functions would add support to the argument that CTCs can exist without a violation of locality or the creation of a grandfather paradox. This is because conditions \eqref{condition2} and \eqref{indep} restrict process functions to avoid such inconsistencies. In order to answer whether non-trivial process functions can exist in the presence of CTCs, we must develop a characterization of process functions for an arbitrary number of regions. In other words, we want to find a way to tell whether a generic function $w: \mathcal{X} \rightarrow \mathcal{A}$ satisfies the fixed point condition \eqref{condition2}.

 Ref.~\cite{Baumeler2019} characterised process functions with up to three regions. For a single region, condition \eqref{indep} requires that the process function has to be a constant: $w(x)~=~a \; \forall \; x$. Bipartite process functions are characterized by three conditions:
\begin{enumerate}
	\item[(i)] $w_1(x_1,x_2)=w_1(x_2)\,$,
	\item[(ii)] $w_2(x_1,x_2)=w_2(x_1)\,$,
	\item[(iii)] at least one of~$w_1(x_2)$ or~$w_2(x_1)$ is constant.
\end{enumerate}
It is clear that (i) and (ii) follow from condition \eqref{indep}, while (iii) follows from condition \eqref{condition2}. As a result, bipartite process functions only allow one-way signalling. 

In order to characterise tripartite process functions, we must consider three distinct regions which we label 1, 2 and 3. This process function has three components $a_1 = w_1(x_2, x_3)$, $a_2 = w_2(x_1, x_3)$ and $a_3 = w_3(x_1, x_2)$. Ref.~\cite{Baumeler2019} proves the following characterisation of tripartite process functions, where the output variable of one region `switches' the direction of signalling between the other two regions.

\begin{theorem}[Tripartite process functions, Theorem 3 in Ref.~\cite{Baumeler2019}]\label{3charact}
	Three functions $w_1:\mathcal{X}_2\times \mathcal{X}_3\rightarrow \mathcal{A}_1$, $w_2:\mathcal{X}_1\times \mathcal{X}_3\rightarrow \mathcal{A}_2$, $w_3:\mathcal{X}_1\times \mathcal{X}_2\rightarrow \mathcal{A}_3$ define a process function if and only if each of the output reduced functions
	\begin{align}
	w^{x_3}(x_1,x_2):=&\left\{w_1(x_2,x_3), w_2(x_1,x_3)\right\},\\
	w^{x_1}(x_2,x_3):=&\left\{w_2(x_1,x_3), w_3(x_1, x_2)\right\},\\
	w^{x_2}(x_1,x_3):=&\left\{w_1(x_2,x_3), w_3(x_1,x_2)\right\}
	\end{align}
is a bipartite process function for every $x_3\in \mathcal{X}_3$, $x_1\in \mathcal{X}_1$, $x_2\in \mathcal{X}_2$ respectively.
\end{theorem}

The properties defined in Theorem \ref{3charact} describe that, for every fixed output of one of the regions, at most one-way signalling is possible between the other two regions. Our goal is to prove a similar characterisation of multipartite process functions in terms of conditional signalling.

\section{Characterization of multipartite process functions}
We are now ready to prove our core result: a characterisation of arbitrary multipartite process functions that generalises Theorem \ref{3charact}. There are in fact two distinct (but equivalent) ways to generalise Theorem \ref{3charact}: given an $N$-partite process function, one can check if all $N-1$-partite functions, obtained by fixing one output, are valid process functions. Alternatively, one can fix all but two outputs, and check if the remaining two regions are at most one-way signalling. Let us start with the first generalisation.

\begin{theorem}[$N$-partite process function]
	\label{fixedthm} 
P[N]: N functions $w_1: \mathcal{X}_{\setminus 1} \rightarrow \mathcal{A}_1$, $w_2: \mathcal{X}_{\setminus 2} \rightarrow \mathcal{A}_2$,  ..., $w_N:\mathcal{X}_{\setminus N} \rightarrow \mathcal{A}_N$, define a process function $w$ if and only if:
\newline

for every $x_i \in \mathcal{X}_i$, the output reduced function $w^{x_i}$ is an N-1 partite process function, for all $i \in 1,2,...,N$.
\end{theorem}

\begin{proof} 
Consider $N$ space time regions, with the $i$th region's input states denoted as $a_i \in \mathcal{A}_i$ and its output states denoted as $x_i \in \mathcal{X}_i$. We know that if $w$ is a process function, then the output reduced function $w^{x_i}$ must also be a valid process function, as we recognise that all output reduced functions are also reduced functions and can therefore apply point (i) of Lemma \ref{lemma}. This proves one direction of the theorem.

In order to complete the proof, we need to prove the converse as well: If $w^{x_i}$ is a valid $N-1$ partite process function for $i \in 1,2,3,....,N$, then $w$ is a valid $N$-partite process function. The proof will proceed by induction. Firstly, we will prove $P[3]$, and then the implication $P[N-1]$ $\Rightarrow$ $P[N]$. $P[3]$ is proven simply by applying Theorem \ref{3charact}. Next, we assume the induction hypothesis is true for the $P[N-1]$ case: for $N-1$ partite function $w$, if $w^{x_i}$ is a valid $N-2$ partite process function for all $i \in 1,2,3,...,N-1$, $\Rightarrow$ $w$ is a valid $N-1$ partite process function.  

For an arbitrary $N$-partite function $w$ we assume $w^{x_i}$ to be a valid $N-1$ partite process function for all $i \in 1,2,3,...,N$. As a result, by applying (i) of Lemma \ref{lemma}, it is easy to see that the reduced function $(w^{x_i})^{f_1}$ is a valid $N-2$ partite process function for $i \neq 1$. Here, our choice of specifying region 1 as opposed to any of the other regions which are not the $i$th region, when fixing the local operation $f_{1}: \mathcal{A}_{1} \rightarrow \mathcal{X}_{1}$, is arbitrary. Now, the order in which operation is fixed does not affect the resulting function, i.e. we have $(w^{x_i})^{f_1} = (w^{f_1})^{x_i}$.

In order to see this, we start by considering the definition of the reduced function. Definition \ref{def1} states that for each component of a function $w$, the reduced function $w^{f_1}$ is given by the composition of $w$ with $f_1$, $w_{j}^{f_{1}}\left(x_{\setminus 1}\right):=w_{j}\left(x_{\setminus 1}, f_{1}\left(w_{1}\left(x_{\setminus 1}\right)\right)\right), 1 \neq j$. If we also fix the output of region $i \in 2,...,N$, to be $x_i$, then the corresponding output reduced function has components 
\[
\begin{split}
    (w^{f_{1}})^{x_i}_{j}\left(x_{\setminus \{1, i\}}\right) :&= w^{f_{1}}_{j}\left(x_{\setminus \{1, i\}},x_{i}\right) \\ &=  w_{j}\left(x_{\setminus 1}, f_{1}\left(w_{1}\left(x_{\setminus 1}\right)\right)\right),
\end{split}
\]
where we recall that $x_{\setminus \{1, i\}} \cup x_{i} = x_{\setminus 1}$.

 Now it is easy to see that if we fix the output $x_i$ before we fix the function $f_1$, we arrive at the same expression:
\[
\begin{split}
    (w_{j}^{x_i})^{f_1}\left(x_{\setminus \{1, i\}}\right) :&= w^{x_i}_{j}\left(x_{\setminus \{1, i\}}, f_{1}\left(w^{x_i}_{1}\left(x_{\setminus \{1, i\}}\right)\right)\right) \\
    &= 
    w_{j}\left(x_{\setminus \{1, i\}},x_{i}, f_{1}\left(w_{1}\left(x_{\setminus \{1, i\}},x_{i}\right)\right)\right)\\
    &= w_{j}\left(x_{\setminus 1}, f_{1}\left(w_{1}\left(x_{\setminus 1}\right)\right)\right).
\end{split}
    \]

Now, we can apply $P[N-1]$ to show that if $(w^{f_1})^{x_i}$ is valid $N-2$ partite process function, then $w^{f_1}$ must be a valid $N-1$ partite process function. Finally, we can apply point (ii) of lemma \ref{lemma} to conclude that $w$ is a valid process function, thus completing the proof.

\end{proof}
Theorem \ref{fixedthm} can be applied as a simple framework to check if a multipartite function is indeed a valid process function. As we will see later, it is easier to apply Theorem \ref{fixedthm} to verify a valid multipartite process function than check that condition \eqref{condition2} is satisfied. The potential for the application of Theorem \ref{fixedthm} as a method for verifying the validity of multipartite process functions can easily be seen by noticing that Theorem \ref{fixedthm} implies that fixing the outputs of all regions except 2, reduces the remaining output reduced process function to a bipartite function.  

\begin{corollary} \label{lemmanew}
N functions $w_1: \mathcal{X}_{\setminus 1} \rightarrow \mathcal{A}_1$, $w_2: \mathcal{X}_{\setminus 2} \rightarrow \mathcal{A}_2$,  ..., $w_N:\mathcal{X}_{\setminus N} \rightarrow \mathcal{A}_N$, define a process function $w$ if and only if
$w^{x_{\setminus\{l,j\}}}$ is a valid bipartite process function for all $l,j \in 1,2,3,...,N$, $l\neq j$, and $N \geq 3$.
\begin{proof}
We know that if $w$ is an $N$-partite process function, then for all $i \in 1,2,...,N$, $w^{x_i}$ must also be a valid process function by applying either point (i) of Lemma \ref{lemma}, or Theorem \ref{fixedthm} directly. We use the same logic to prove that for all $i,k \in 1,2,...,N$, $i\neq k$, $(w^{x_i})^{x_k}$ must be a valid $N-2$ partite process function. Repeating the argument until we have fixed the output of all regions except two proves that if $w$ is an $N$-partite process function, then $w^{x_{\setminus\{l,j\}}}$ is a valid bipartite process function, proving one direction of the corollary. 

In order to prove the converse, we begin by noting that we can write an arbitrary bipartite function as an output reduced tripartite process function: $w^{x_{\setminus \{j, l\}}} = \left(w^{x_{\setminus \{i, j, l\}}}\right)^{x_i}$ for $i\neq j\neq l$. If, for all distinct $j,l \in 1,2,3,...N$,  $w^{x_{\setminus \{j,l\}}}$ is a valid bipartite process function, then by Theorem \ref{fixedthm}, for all distinct  $i,j,l \in 1,2,3,...N$, $w^{x_{\setminus \{i, j, l\}}}$ must be a valid tripartite process function. We can repeat this argument in order to conclude that for all distinct $i,j,k,l \in 1,2,3,...N$, $w^{x_{\setminus \{i, j, k, l\}}}$ must be a valid quadripartite process function. We can keep applying the same argument until we conclude that $w$ is a valid $N$-partite process function, thus proving the reverse direction of the corollary and hence concluding the proof.
\end{proof}
\end{corollary}

 Corollary \ref{lemmanew} explicitly demonstrates the condition that fixing the output of all regions except two arbitrarily picked $l,j \in 1,2,...,N$ determines the direction of signalling between regions $l$ and $j$, and hence the remaining output reduced process function is a bipartite process function:
\[ w^{x_{\setminus \{l,j\}}}(x_l,x_j) := \{w_l(x_j, x_{\setminus \{l,j\}}), w_j(x_l, x_{\setminus \{l,j\}}) \}, \]
where at least one of the two component functions $w^{x_{\setminus \{l,j\}}}_l$, $w^{x_{\setminus \{l,j\}}}_j$ is a constant. The validity of a multipartite process function can be checked by ensuring that, for all $l, j \in 1,2,3,...,N$, $w^{x_{\setminus \{l,j\}}}$ is a valid bipartite process function.

Theorem \ref{fixedthm} and subsequently Corollary \ref{lemmanew} demonstrates that multipartite process functions can at most be conditionally one-way signalling between any pair of regions. In other words, fixing the output of all regions except two, allows at most one-way signalling between the two remaining regions.

\section{Examples}
The above characterisation of process functions allows us to consider specific examples that cannot occur in an ordinary, causally ordered spacetime. An example of a such a process function in three spacetime regions was first presented in Ref.~\cite{Baumeler2016}. This tripartite process function can easily be extended to a quadripartite process function through the addition of a fourth party either in the past or future of the other three parties. However, the existence of the fourth party in the process function does not require the presence of CTCs. In the case where the fourth party is in the future of the other three parties, this simply corresponds to a fourth region where causal order exists from the other three parties to the fourth party. As a result, there is significant motivation to find quadripartite process functions incompatible with causal order between any subsets of parties (this is analogous to the ``genuinely multipartite non-causal correlations'' studied for quantum processes \cite{Abbott2017genuinely}). 

Here, we present examples of such quadripartite process functions. Consider four parties in local regions of spacetime in the presence of CTCs. In this analysis, we simplify the classical state spaces $\mathcal{A}_i$ and $\mathcal{X}_i$ to be binary state spaces. A physical example of what these binary variables may represent is that a 0 or 1 could correspond to the existence or non-existence of a particle on the boundary of the spacetime region. As discussed previously, we do not consider any specific spacetime or dynamical law, but simply investigate what processes are logically possible. 

We define input variables $a_1$, $a_2$, $a_3$, $a_4 \in \{0, 1\}$. We define output variables $x_1$, $x_2$, $x_3$, $x_4 \in \{0, 1\}$. We define the binary addition operator $a \oplus b$: $a,b \in \{0, 1\} \rightarrow \{0, 1\} $ as

\begin{equation} \label{binary}
   a \oplus b = \left\{
\begin{array}{ll}
      0, & a = b \\
      1, & a \neq b. \\
\end{array} 
\right. 
\end{equation}
Using the above notation, we define the quadripartite process function $w$: $(x_1$, $x_2$, $x_3$, $x_4)\rightarrow$ $(a_1$, $a_2$, $a_3$, $a_4)$ as

\begin{equation} \label{fourpartite}
    \begin{array}{ll}
      a_1 = x_4(x_2 \oplus 1)(x_3 \oplus 1) \\
      a_2 = x_1(x_4 \oplus 1)(x_3 \oplus 1) \\
      a_3 = x_2(x_1 \oplus 1)(x_4 \oplus 1) \\
      a_4 = x_3(x_2 \oplus 1)(x_1 \oplus 1).
\end{array} 
\end{equation}
Applying either Theorem \ref{fixedthm} or Corollary \ref{lemmanew}, one can check that this is a valid process function. Equation \eqref{fourpartite} defines a process function in which the input of each region depends non-trivially on the output of the other three regions. In this process, the output of two regions sets the direction of signalling between the other two. For example, Table \ref{table1} displays the resulting inputs of regions $3$ and $4$ for all the possible combinations of the outputs of regions $1$ and $2$.

\onecolumngrid
\vspace{.3 cm}
\begin{center}
\begin{table}[ht!] 
\begin{tabular}{|l|l|l|l|l|}
\hline
Output of  &Output of  & Input of  & Input of & Direction of \\
region 1 ($x_1$)& region 2 ($x_2$) &region 3 ($a_3$)  &   region 4 ($a_4$)& signalling \\ \hline
   0  & 0 & 0  & $x_3$ & 3 signals to 4 \\ \hline
  0   & 1 & $x_4 \oplus 1 $ & 0  & 4 signals to 3 \\ \hline
   1  & 0 &  0& 0 & No signalling \\ \hline
     1  & 1 & 0 & 0 & No signalling \\ \hline
\end{tabular}
\caption{\label{table1} 
Inputs of region 3 and region 4 (denoted $a_3$ and $a_4$ respectively), for every possible combinations of the outputs of region 1 and region 2 (denoted $x_1$ and $x_2$ respectively). The displayed signalling structure is true regardless of which two regions we choose to fix the outputs, due to the symmetry between different components of equation \eqref{fourpartite}.} 
\end{table}
\end{center}
\twocolumngrid

Here it is clear that, depending on the choice of outputs for an observer in region 1 and another observer in region 2, the communication between regions 3 and 4 can either be non-existent (neither region can signal to each other), or at most one-way signalling. As a result, equation~\eqref{fourpartite} is characterised by conditional signalling. For example, if the outputs of regions 1 and 2 are chosen to be $x_1 = 1$, $x_2 = 0$ respectively, then neither one of the observers in regions 3 and 4 can signal to each other. However, if the output of region 1 and region 2 are chosen to be $x_1 = 0$ and $x_2 = 0$ respectively, then an observer in region 3 can signal to an observer in region 4. Crucially, there exist combinations of outputs such that each observer can signal to any observer in another region. 

Our analysis of the signalling between regions in the output reduced process function is demonstrating that when we consider a subset composed of just two regions, there can only be at most one-way signalling between these two regions. This ensures that our assumption of locality and the absence of a grandfather paradox is satisfied, which we have formalised in condition~\eqref{condition2}. However, we stress that in our quadripartite examples, each region can signal to any other for some value of the remaining regions' output. This is not possible in a spacetime without CTCs. Indeed, in the absence of CTCs, there would be one region that is either in the causal past or space-like from each of the other regions. In this case no other region would be able to signal to it, regardless of the output of the other regions.

It was found in Ref.~\cite{Baumeler2016} that, up to relabelling of parties or of inputs/outputs, there exists only one non-trivial tripartite process function with binary inputs and outputs that is compatible with the presence of CTCs and incompatible with any causal order. In other words, all non-trivial tripartite functions compatible with the presence of CTCs are equivalent after a relabelling of parties or of states\footnote{This is not strictly true beyond binary state spaces (for example, the continuous variables example presented in Ref.~\cite{Baumeler2019} cannot be relabelled to be a three bits function). However, Theorem \ref{3charact} implies that all non-causally-ordered tripartite process functions have the same causal structure, in the sense that the output space of each party can be divided in two subsets, with states within each subset corresponding to a fixed direction of signalling between the two other parties.}. However, this is not true for quadripartite process functions. There exists many quadripartite process functions that are not related to one another by relabelling of party or of states. An example of a non-trivial quadripartite process function not related to the one given in equation \eqref{fourpartite} up to relabelling is
\begin{equation} \label{fourpartite2}
    \begin{array}{ll}
      a_1 = x_2 (x_3 \oplus x_4) \\
      a_2 = x_3(x_4(x_1 \oplus 1) \oplus 1) \\
      a_3 = x_4 (x_1 \oplus 1 )(x_2 \oplus 1) \\
      a_4 = x_1(x_2 \oplus 1)(x_3 \oplus 1).
\end{array} 
\end{equation}
It is easy to see that the signalling structure between the four components of equation \eqref{fourpartite2} is different from the signalling structure between the four components of equation \eqref{fourpartite}. For example, if we fix the outputs of regions 2 and 3 to be either $x_2 = 1$, $x_3 = 0$ or $x_2 = 1$, $x_3 = 1$, then an observer in region 4 can signal to an observer in region 1. As a result, equation \eqref{fourpartite2} has produced a scenario in which there are two distinct choices for outputs for two components of the process function which result in the same signalling direction (region 4 to region 1). This scenario does not exist for any choices for the outputs of two distinct components of the process function described by equation \eqref{fourpartite}.

We have shown that there exist distinct non-trivial quadripartite process functions that are compatible with the presence of CTCs and incompatible with any causal order. A numerical search for other quadripartite process functions satisfying Theorem \ref{fixedthm} revealed a large number of non-equivalent quadripartite process functions. In this paper we have presented two examples of such process functions. In comparison to tripartite process functions, quadripartite process functions allow a greater variety in the ways different regions can communicate without causal order in the presence of CTCs.

Of course, one of the main question is whether the abstract examples we have found can be realised as concrete physical systems in some appropriate space-time geometry. Although a full answer is beyond the scope of this work, it is always possible to construct a simplified model that implements any process function. Indeed, it was proven in Ref.~\cite{Baumeler2019} that every process function can be extended to a reversible one, possibly adding extra regions and degrees of freedom. This, in turn, implies that there exist reversible physical processes implementing the function. A prominent example is the  ``billiard ball model'' of computation \cite{Fredkin1982}. In this model, $0$ ($1$) represents the absence (presence) of a billiard ball. Through appropriate collisions, which might involve reflecting walls or additional balls, it is possible to implement any functional relation between initial and final state. To turn this into a time-travelling model, it is sufficient to embed the computation into a spacetime with CTCs, such that the output of the computation is mapped identically to the input of the respective local regions, which are themselves in the past of the collision region. An example with four wormholes, based on the spacetimes in Ref.~\cite{Echeverria:1991ko}, is depicted in Fig.~\ref{billard}.

If the dynamics outside the four local regions in Fig.~\ref{billard} is determined by acausal process functions, such as Eq.~\eqref{fourpartite} or Eq.~\eqref{fourpartite2}, then the agents can communicate with each other, while verifying to be in the past, present and future of each other. As these functions satisfy the self consistency condition \eqref{condition2}, they guarantee that, for every choice of local operations, there is a consistent solution. Therefore, the agents can signal to each other without causal order and without any logical inconsistency. It is an open question how generic is this situation and what are the spacetimes and physical systems for which non-trivial, self-consistent time travel is possible.

\begin{figure}[h] 
\begin{center}
\includegraphics[scale  = 0.36]{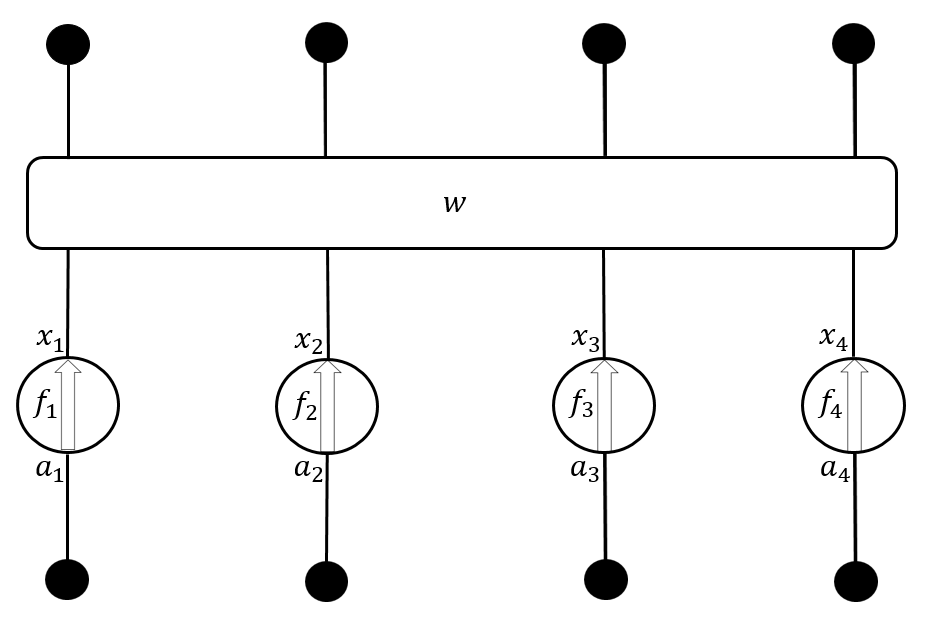}
    \caption{\label{billard} Spatial diagram of a physical realisation of an acausal process function. The example is that of a wormhole spacetime, similar to Ref.~\cite{Echeverria:1991ko}: the metric is Minkowski everywhere, except for eight cylindrical regions (the wormholes' `mouths', appearing as spheres in a spatial section which are the black coloured circles in the picture), which are pairwise identified. A time delay between the mouths causes a world-line entering an upper mouth to exit the corresponding lower mouth at an earlier time in Minkowski coordinates, forming CTCs. Just above each of the lower mouths is a local region, labelled $i=1,\dots 4$, where an agent can either receive ($a_i = 1$) or not receive ($a_i = 0$) a billiard ball along a prescribed world-line (the picture shows the projection of the world-lines on a time slice containing the local regions). The agents can then choose whether or not to send a billiard ball out of the local region. This intervention is represented by the function $f_i$, while the output is $x_i=1,0$. In the space between the regions and the upper mouths, there is a ``billiard ball computer'' \cite{Fredkin1982}, which implements the process function $w$. The properties of the process function ensure that a unique solution for the presence or absence of a ball at each point of the considered worldlines exists for each choice of local operations.}
\end{center}
\end{figure}


\section{Conclusions}
We have developed a characterisation of deterministic processes in the presence of CTCs for an arbitrary number of localised regions. Our proofs have demonstrated that non-trivial time travel between multiple regions is consistent with the absence of a logical paradox as long as once the outputs of all but two regions are fixed, at most one-way signalling is possible.

The most significant result of our work is our discovery of distinct non-trivial quadripartite process functions which are compatible with the presence of CTCs. This demonstrates that when multiple local regions communicate with each other in the presence of CTCs, there is a broad range of communication scenarios which still allow freedom of choice for observers in each region without the development of a logical inconsistency such as a grandfather paradox. The range of distinct communication scenarios which are consistent with the presence of CTCs proves that the way CTCs allow multiple observers in distinct regions to communicate is not overly restricted by a conflict between locality, freedom of choice, and logical consistency. As a result, we have demonstrated that there is a range of scenarios in which multiple observers can communicate without causal order in a classical framework. Our results are derived in an abstract framework, that does not depend on the details of the dynamics or of the space-time geometry. Further studies will be necessary to find genuine physical scenarios realising the acausal processes we have discovered.

\section{Code Availability}
The code is available at \url{https://github.com/gtobarSU/Classical-Process-Matrix-Search-Code.git}.

	{\bf Acknowledgments.}
We thank S.\ Staines for collaboration during an undergraduate winter research project at the University of Queensland. F.C.\ acknowledges support through an Australian Research Council Discovery Early Career Researcher Award (DE170100712).  We acknowledge the traditional owners of the land on which the University of Queensland is situated, the Turrbal and Jagera people.

\bibliography{refs}

\end{document}